\documentclass{article}

\usepackage{amsmath}
\usepackage{amssymb}
\usepackage{amstext}
\usepackage{amsthm}
\usepackage{mathtools}
\usepackage{enumerate}
\usepackage{url}
\usepackage{dsfont}
\usepackage[pdftex]{graphicx} 
\usepackage[usenames,dvipsnames]{color}
\usepackage{upgreek}
\usepackage{float}
\usepackage{algorithm}
\usepackage[noend]{algorithmic}
\usepackage[margin=1 in]{geometry}
\usepackage{subcaption} 
\usepackage{hyperref} 
\theoremstyle{definition}

\newtheorem{problem}{Problem}

\newtheorem{claim}{Claim}
\newtheorem{theorem}{Theorem}[]
\newtheorem{lemma}[theorem]{Lemma}
\newtheorem{corollary}[theorem]{Corollary}

\newcommand{\ba}{\[\begin{aligned}}
\newcommand{\ea}{\end{aligned}\]}
\newcommand{\bi}{\begin{itemize}}
\newcommand{\ei}{\end{itemize}}
\newcommand{\bb}{\backslash}
\newcommand{\M}{\mathcal{M}}
\newcommand{\N}{\mathbb{N}}
\newcommand{\R}{\mathbb{R}}
\begin{document}

\title{Two-sided popular matchings in bipartite graphs \\ with forbidden/forced elements and weights}

\author{Yuri Faenza, Vladlena Powers, and Xingyu Zhang\thanks{email: yf2414, vp2342, xz2464@columbia.edu.}\\ IEOR, Columbia University}

\date{}

\maketitle

\begin{abstract} 	Two-sided popular matchings in bipartite graphs are a well-known generalization of stable matchings in the marriage setting, and they are especially relevant when preference lists are incomplete. In this case, the cardinality of a stable matching can be as small as half the size of a maximum matching. Popular matchings allow for assignments of larger size while still guaranteeing a certain fairness condition. In fact, stable matchings are popular matchings of minimum size, and a maximum size popular matching can be as large as twice the size of a(ny) stable matching in a given instance. The structure of popular matchings seems to be more complex -- and currently less understood -- than that of stable matchings. 
	
In this paper, we focus on three optimization problems related to popular matchings.  First, we give a granular analysis of the complexity of \emph{popular matching with forbidden and forced elements} problems, thus complementing results from~\cite{The popular edge problem}. In particular, we show that deciding whether there exists a popular matching with (resp. without) two given edges is NP-Hard. This implies that finding a popular matching of maximum (resp. minimum) weight is NP-Hard and, even if all weights are 
nonnegative, inapproximable up to a factor $1/2$ (resp. up to any factor), unless P=NP. A decomposition theorem from~\cite{The popular edge problem} can be employed to give a $1/2$ approximation to the maximum weighted popular matching problem with nonnegative weights, thus completely settling the complexity of those problems. Those problems were posed as open questions by a number of papers in the area (see e.g.~\cite{Cseh on popular matchings, The popular edge problem, Popularity mixed matchings and self-duality, Popular half-integral matchings, Manlove}).

	% while computing a \emph{maximum weight popular matching} (\texttt{mwp}) is NP-Hard and not approximable better than a factor $1/2$. Whether \texttt{mwp} could be solved in polynomial time was arguably of the main open questions in the area, and mentioned as such in~\cite{Cseh on popular matchings},~\cite{Popularity mixed matchings and self-duality},~\cite{Popular half-integral matchings} {\color{red}...find all papers that mention this as an open question}. More generally, we investigate the complexity of different versions of the \emph{popular matching with forbidden and forced elements} problem, thus complementing results from~\cite{The popular edge problem}. 

	\end{abstract}	

\section{Introduction}

Matching problems lie at the intersection of different areas of mathematics, computer science, and economics, and at the core of each of those. In discrete mathematics, they have been studied at least since Petersen (see e.g.~the historical survey in~\cite{Lov Plum}); in discrete optimization, Edmonds' maximum weighted matching algorithm is one the earliest and most remarkable example of non-trivial polynomial-time algorithm; in game theory, network bargaining problems models are tightly connected to maximum matchings and their linear programming formulations~\cite{KleTa}; in two-sided markets, Gale-Shapley's model~\cite{GaSha} has been widely used and generalized to assign e.g.~students to schools and interns to hospitals. 

Gale-Shapley's model assumes that the input graph is bipartite, and that each node is endowed with a strict preference list over the set of its neighbors. The goal is to find a matching that respects a certain concept of fairness called \emph{stability}. This model has been generalized in many ways. On one hand, one can change the structure of the \emph{input}, assuming e.g.~the presence of ties in the preference lists, or of quotas in one of the two sides of the bipartition, or that the preference pattern is given by more complex choice functions (see e.g.~\cite{Manlove} for a collection of these extensions). On the other hand, one can change the requirements of the \emph{output}, i.e. ask that the matching we produce satisfies properties other than stability. For instance, relaxing the stability condition to \emph{popularity} allows to overcome one of the main drawbacks of stable matchings: the fact that two individuals (a \emph{blocking pair}) can prevent the matching from being much larger. Roughly speaking, a matching $M$ is \emph{popular} if, for every other matching $M'$, the number of nodes that prefer $M$ to $M'$ are at least as many as the number of nodes that prefer $M'$ to $M$ (this model is called \emph{two-sided}; formal definitions are given in Section~\ref{sec:two}). One can show that stable matchings are popular matchings of minimum size, and the size of a popular matching can be twice as much as the size of a stable matching (recall that all stable matchings are maximal, and a maximal matching has size at least half the size of a maximum matching). Hence, popularity allows for matchings of larger size while still guaranteeing a certain fairness condition.

Surprisingly, most of what is known about popular matchings derives in a way or another from stable matchings. For instance, Kavitha~\cite{A size-popularity tradeoff in the stable marriage problem}  showed that a maximum size popular matching can be found efficiently by a combination of repeated applications of the \emph{Gale-Shapley's} algorithm and \emph{promotion}. Cseh and Kavitha~\cite{The popular edge problem} showed that a pair of nodes are matched together in at least a popular matching if and only if they are matched together either in some stable matching, or in some \emph{dominant} matching, with the latter being (certain) popular matchings of maximum size, and the image of stable matchings in another graph under a linear map. This fact implies that the \emph{popular edge problem} can be solved in polynomial time (see Section~\ref{sec:def} for definitions). In~\cite{The popular edge problem}, it is asked whether their results can be extended to find a popular matching that contains two given edges. Other open questions in the area involve finding a popular matching of maximum (a generalization of the popular edge problem) or minimum (and its special case, the \emph{unpopular edge problem}) weights, in particular when weights are nonnegative, see e.g.~\cite{The popular edge problem, Cseh on popular matchings, Popularity mixed matchings and self-duality, Popular half-integral matchings}. %Other important open questions in the area deal with the existence of efficient algorithms to detect if a non-bipartite graph has a popular matching and {\color{red}what else?}

\smallskip

\noindent {\bf Our contribution.} In this paper, we settle the complexity of all those problems by showing that it is NP-Complete to decide whether there exists a popular matching that (does not) contain two given edges. More generally, we give a very granular analysis of the complexity of the \emph{popular matching with forbidden and forced elements} problem, i.e. deciding whether a popular matching exists when we force it (not) to have certain edges or vertices. 

From our reduction, it follows that the problem of finding a popular matching of \emph{minimum} weight when weights on edges are nonnegative cannot be approximated up to any factor (one can indeed define the objective function so that it is NP-Complete to decide if the optimum is $0$). If instead, still under the nonnegativity assumption on the weights, we want to find a popular matching of \emph{maximum} weight, the problem becomes inapproximable to a factor better than $1/2$. It follows from a decomposition theorem for popular matchings given in~\cite{The popular edge problem} that a $1/2$ approximation can be achieved in polynomial time by taking the matching of maximum weight from the set given by dominant and stable. This shows that our result is tight. Our findings suggest that stable matchings and their closely connected counterpart dominant matchings seem to be the only tractable subclasses of popular matchings.

\smallskip

\noindent {\bf Organization of the paper.} The paper is organized as follows. In Section~\ref{sec:two}, we formally define the main problems under consideration and state our results. In Section~\ref{sec:def}, we introduce suitable notation and recall results that will be used in proofs. The main technical contribution is given in Section~\ref{sec:reduction}, where a reduction from monotone $3$-SAT to the existence of certain popular matchings is given. In Section~\ref{sec:main-results}, we combine basic facts showed in Section~\ref{sec:def} and the reduction from Section~\ref{sec:reduction} to investigate the complexity of popular matching with forbidden and forced elements problems. We conclude with some further NP-Hardness results in Section~\ref{sec:last}.

\section{Problems and statement of main results}\label{sec:two}

Let  $G=(A\cup B, E)$ be a, simple connected bipartite graph. Since all graphs under consideration are simple and non-directed, we denote the edge between nodes $u$ and $v$ as $uv$. 
For each vertex $v \in V$, let $\Gamma(v)$ denote the neighbors of $v$ in $G$ and let $<_v$ be a strict total order over $\Gamma(v)$. If $u <_v w$, we say that $v$ \emph{prefers} $u$ to $w$. We write $u \leq _v w$ if $w=u$ or $u <_v w$, and  we say that $v$ \emph{weakly prefers} $u$ to $w$.  The pair $(G, <)$ is called a \emph{preference system}, where $<$ is the collection of strict total orders $<_v$ for $v \in V$. 

A \emph{matching} of $(G,<)$ is a subset $M$ of the edges of $G$ such that, for each node of $G$, there is at most an edge of $G$ incident to it. If there exists $e \in M$ such that $v \in e$, we say that $v$ is \emph{matched} or \emph{covered} by $M$. We say that $v$ is \emph{unmatched}, or \emph{exposed by $M$}, or \emph{$M$-exposed} otherwise. We denote by $V(M)$ the set of nodes matched by $M$. %Given subsets $A' \subseteq A$ and $B' \subseteq B$, we let $E(A',B')$ denote the set of edges of $G$ that have one endpoint in $A'$ and the other in $B'$.  
A vertex $v \in A\cup B$ is said to \emph{prefer} matching $M$ to matching $M'$ if either $v$ is matched in $M$ and unmatched in $M'$, or $v$ is matched in both and $v$ prefers its partner in matching $M$ to its partner in matching $M'$. For matchings $M$ and $M'$, let $\phi(M,M')$ be the number of vertices that prefer $M$ to $M'$. If $\phi(M',M) > \phi(M,M')$, then we say $M'$ is \emph{more popular than} $M$. We say a matching $M$ is \emph{popular} if there is no matching that is more popular than $M$; in other words, $\phi(M,M') \geq \phi(M',M)$ for all matchings $M'$ in $(G,<)$.

%Given two matchings $M,M'$, we say that $v$ prefers $M$ to $M'$ if $v$ is matched in $M$ but not in $M'$, or if $v$ prefers $M(v)$ to $M'(v)$. A matching of $(G,<)$ is called popular if, for every other matching $M'$, the number of nodes that prefer $M$ to $M'$ is greater or equal to the number of nodes that prefer $M'$ to $M$.

In this paper, we mostly consider the following three problems.

\begin{problem}[Popular matching with forbidden and forced elements problem -- \texttt{pmffe}]\label{pb:main problem} 
	{\bf Given:} A preference system $(G,<)$, $U^{in},U^{out} \subset V(G)$, $F^{in},F^{out} \subset E(G)$, such that $U^{in}\cap U^{out}=\emptyset $, $F^{in}\cap F^{out}=\emptyset$, and if $uv \in F^{in}\cup F^{out}$ then $u,v \notin U^{in}\cup U^{out}$. \textbf{Find}: A popular matching of $(G,<)$ that contains all vertices of $U^{in}$, all edges of $F^{in}$, no vertices of $U^{out}$ and no edges of $F^{out}$, or conclude that such matching does not exist.
\end{problem}

\begin{problem}[Maximum weight popular matching with nonnegative weights -- \texttt{mwp}]\label{pb:mwp} 
	{\bf Given:} A preference system $(G,<)$, weights $w:E(G)\rightarrow \R_+$. \textbf{Find}: A popular matching $M$ of $(G,<)$ that maximizes $w(M)$.
\end{problem}

\begin{problem}[Minimum weight popular matching with nonnegative weights -- \texttt{miwp}]\label{pb:mwp} 
	{\bf Given:} A preference system $(G,<)$, weights $w:E(G)\rightarrow \R_+$. \textbf{Find}: A popular matching $M$ of $(G,<)$ that minimizes $w(M)$.
\end{problem}

Our first result is a complete characterization of the complexity of of \texttt{pmffe} as a function of the cardinality of sets $U^{in}$, $U^{out}$, $F^{in}$, $F^{out}$.

\begin{theorem}
	\label{thr: 3SAT-popular}
	If $U^{out}=F^{in}=F^{out} = \emptyset$ or $U^{in}=F^{in}=F^{out} = \emptyset$, or ${|U^{in}\cup U^{out}\cup F^{in}\cup F^{out}|\leq 1}$, then \texttt{pmffe} can be solved in polynomial time. Otherwise, for each other set of fixed values for $|U^{in}|$,$|U^{out}|$,$|F^{in}|$,$|F^{out}|$, \texttt{pmffe}  is NP-Complete.
\end{theorem}

We then investigate consequences of Theorem~\ref{thr: 3SAT-popular} for some optimization problems over the set of popular matchings. It easily follows that the problems of maximizing/minimizing \emph{any} linear function is inapproximable up to any factor, unless P=NP (see Section~\ref{sec:last} for details). Theorem~\ref{thr: 3SAT-popular} gives hardness and tight inapproximability results for \texttt{mwp} and \texttt{miwp}.

\begin{theorem}\label{thr:mwm} \texttt{miwp} and \texttt{mwp} are NP-Hard. Unless P=NP, \texttt{miwp} cannot be approximated in polynomial time up to any factor,  and \texttt{mwp} cannot be approximated in polynomial time better than a factor $\frac{1}{2}$. On the other hand, there is a polynomial-time algorithm that computes a $\frac{1}{2}$-approximation to \texttt{mwp}.
	%There is no function $f : \N \rightarrow \R_{>0}$ such that there exists a polynomial time algorithm that, given as input an instance of the maximum weight popular matching problem with $n$ nodes, outputs a popular matching of weight at least $f(n)$ times the weight of a popular matching of maximum weight.
\end{theorem}

\section{Definitions and basic facts from the theory of popular matchings}\label{sec:def}

Given $(G,<)$, we often represent $<$ via a \emph{preference pattern}, i.e., a function $L$ that maps ordered pairs $(u,v)$ with $uv \in E$ to the position of $v$ in $<_u$. If $L(u,v)=1$, we say that $v$ is the \emph{favorite partner} of $u$. An \emph{incomplete preference pattern} $I$ for a bipartite graph $G$ is a map $I: \{(u,v) :  uv \in E(G)\} \rightarrow \N \cup \{\infty, \emptyset\}$, such that $I(u,v)=I(u,x) \text{ for some } u \neq x \neq v \in V$ implies $I(u,v)=I(u,x)=\emptyset$, and $I(u,v)\in \N$ implies $I(u,v)\leq |\Gamma(u)|$. We say that a preference pattern $L$ for a supergraph $G'$ of $G$ \emph{agrees} with an incomplete preference pattern $I$ if, for all $uv \in E(G)$ with $I(u,v)\neq \emptyset$, we have: $L(u,v)=I(u,v)$ whenever $I(u,v) \in \N$ and $L(u,v)=|\Gamma_{G'}(u)|$ whenever $I(u,v)=\infty$  (here $\Gamma_{G'}(u)$ is the neighborhood of $u$ in $G'$).

Fix a matching $M$ of $(G,<)$. An $M$-alternating path (resp. cycle) in $G$ is a path (resp. cycle) whose edges are alternatively in $M$ and in $E(G)\setminus M$. We morover associate labels to edges from $E\setminus M$ as follows: an edge $uv$ is $(-,-)$ if both $u$ and $v$ prefer their respective partners in $M$ to each other; $uv = (+,+)$ if $u$ and $v$ prefer each other to their partners in $M$; $uv = (+,-)$ if $u$ prefers $v$ to its partner in $M$ and $v$ prefers its partner in $M$ to $u$. To denote labels of edge $uv$, sometimes we write \emph{$uv$ is a $(+,+)$ edge}, and similarly for the other cases. If we want to stress that $u$ prefers his current partner to the other endpoint of the edge, we write that e.g.~$uv$ is a $(-,+)$ edge.  The graph $G_M$ is defined as the subgraph of $G$ obtained by deleting edges that are labeled $(-,-)$. Observe that $M$ is also a matching of $G_M$, hence definitions of $M$-alternating path and cycles apply in $G_M$ as well. These definitions can be used to obtain a characterization of popular matchings in terms of forbidden substructures of $G_M$~\cite{Popular matching in the stable marriage problem}. We consider path and cycles as ordered sets of nodes. For a path (resp. cycle) $P$ and an edge $e$, we say that \emph{$P$ contains $e$} if $e=uv$ for two nodes $u,v$ that are consecutive in $P$ (not necessarily in this order).

\begin{theorem}\label{thr:characterize-popular}
	Let $(G,<)$ be a preference system, and $M$ a matching of $G$. Then $M$ is popular if and only if $G_M$ does not contain any of the following:
	\begin{enumerate}[(i)]
		\item an $M$-alternating cycle that contains a $(+,+)$ edge.
		\item an $M$-alternating path that contains two distinct $(+,+)$ edges.
		\item an $M$-alternating path starting from an $M$-exposed node that contains a $(+,+)$ edge. 
	\end{enumerate}
\end{theorem}

If the matching $M$ under consideration is clear, we omit mentioning that labels are wrt $M$ and paths, cycles etc.~occur in the graph $G_M$: we simply say e.g.~that $uv$ is a $(+,+)$ edge and $P$ is an $M$-alternating path with a $(+,+)$ edge. In particular, unless stated otherwise, when talking about $M$-alternating paths or cycles, we always mean in the graph $G_M$ (hence, after $(-,-)$ edges have been removed).

We say a matching $M$ \emph{defeats} a matching $M'$ (and $M'$ is defeated by $M$) if either of these two conditions holds:
\begin{itemize}
	\item[(i)] $M$ is more popular than $M'$, i.e.,  $\phi(M,M') > \phi(M',M)$;
	\item[(ii)]  $\phi(M,M') = \phi(M',M)$ and $|M|>|M'|$.
\end{itemize}

A \emph{dominant} matching is defined to be a matching that is never defeated. A dominant matching is by definition also popular. The following characterization of dominant matchings appeared in~\cite{The popular edge problem}, where it is also shown that a dominant matching always exists.

\begin{theorem}
Let $M$ be a popular matching. $M$ is dominant if and only if there is no $M$-augmenting path in $G_M$.
\end{theorem}
The following fact is attributed in~\cite{Cseh on popular matchings} to~\cite{On the structure of
	popular matchings in the stable marriage problem}. We give a proof here for completeness. 

\begin{lemma}\label{lem:node-containement}
	Let $(G,<)$ be a preference system, and $M$ (resp. $S$, $D$) be a popular matching (resp. popular matching of minimum size, popular matching of maximum size) in $(G,<)$. Then $V(S)\subseteq V(M)\subseteq V(D)$.
\end{lemma}

\begin{proof}
	The fact that stable matchings are popular matchings of minimum size and that $V(S)\subseteq V(M)$ whenever $S$ is stable and $M$ is popular are known, see e.g.~\cite{A size-popularity tradeoff in the stable marriage problem}.  Hence, for any popular matching $S'$ of minimum size, we have $V(S')\supseteq V(S)$ and $|S'|=|S|$, which implies $V(S')=V(S)\subseteq V(M)$ for any popular matching $M$, settling the first part. Now suppose that $M$ again is popular, $D$ is a dominant matching, and there exists $v \in V(M)\setminus V(D)$. Then, $G(M\triangle D)$ contains an alternating path $P$ starting at $v$ with an edge of $M$. Since both matchings are popular, the number of nodes that prefer $M$ to $D$ in $P$ are exactly the same number of nodes that prefer $D$ to $M$ in $P$. In particular, $P$ has an even number of nodes and is therefore $D$-augmenting in $G$. Now label edges of $P\setminus D$ with $+,-$ as discussed at the beginning of the section. Note that the number of $(+,+)$ edges and $(-,-)$ edges in $P$ coincide, otherwise either $D$ or $M$ is not popular. We consider three cases: if there is no $(+,+)$ edge, hence no $(-,-)$, the path is $D$-augmenting in $G_D$, contradicting the fact that $D$ is dominant. Now suppose there is a $(+,+)$ edge $e$. Since the number of $(+,+)$ and $(-,-)$ edges in $P$ coincide, one of the following two facts happen: either we can assume wlog that $e$ is traversed in $P$ before any $(-,-)$ edge, or there is a subpath of $P$ that contains two $(+,+)$ edges and no $(-,-)$ edge. In both circumstances, we contradict Theorem~\ref{thr:characterize-popular}: condition (ii) in the second case, and condition (iii) in the first.
\end{proof}

\section{The reduction}\label{sec:reduction}

\subsection{From a restricted $3$-SAT formula to $G(\psi)$}

A 3-SAT instance $\psi$ is a logic formula given by the conjunction of a collection of ${\cal K}$ clauses, each of which is the disjunction of at most $3$ literals. We denote the set of variables by ${\cal X}=\{x_1,\dots, x_n\}$, the set of clauses by ${\cal C}=\{c_1,\dots,c_m\}$, and the set of literals by ${\cal L}=\{\ell_1,\dots,\ell_{p}\}$, with $p\leq 3m$. Literals of the form $x_i$ are called \emph{positive} (\emph{occurrences of variable $i$}), while those of the form $\neg x_i$ are called \emph{negative} (\emph{occurrences of variable $i$}). In this paper, we restrict our attention to monotone 3-SAT, which is the version of 3-SAT where each clause consists of only positive or only negative literals. This is also known to be NP-hard, since as observed in~\cite{3sat monotone}, a reduction from 3-SAT is achieved using the following mapping: replace clause $c=(x_1 \lor \neg x_2 \lor \neg x_3)$ by $(x_1 \lor x_c)\land(\neg x_c \lor \neg x_2 \lor \neg x_3)$ and replace clause $c=(x_1 \lor x_2 \lor \neg x_3)$ by $(x_1 \lor x_2 \lor x_c)\land(\neg x_c \lor \neg x_3)$, where for each clause $c$, $x_c$ is a new variable. The other clauses remain unchanged.

 We refer therefore to clauses of a restriced 3-SAT instance as positive clauses and negative clauses accordingly. From now on, we assume without loss of generality that monotone 3-SAT instances under consideration do not consist of only positive (resp. negative) literals, or that a variable only appears as a positive (resp. negative) literal.

\begin{figure}
\begin{center}
\includegraphics[scale=1]{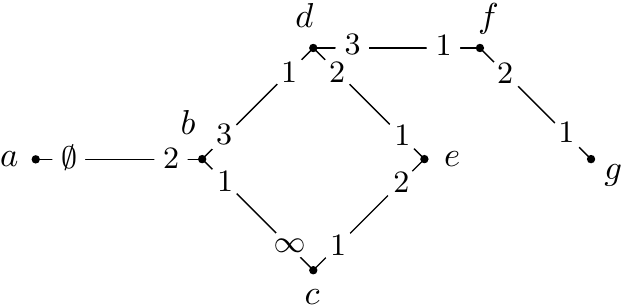}\hspace{1cm}\includegraphics[scale=1]{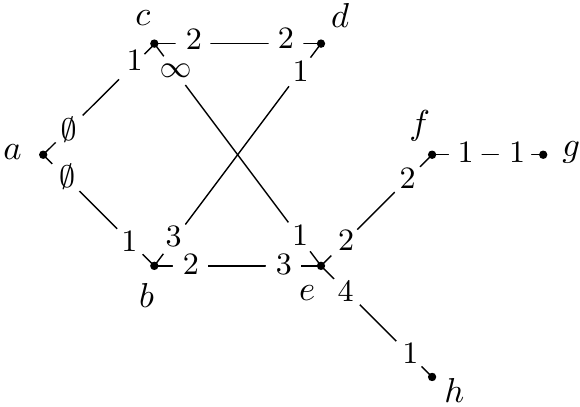}
\end{center}	\caption{The gadget associated to the positive (left) and negative (right) literals, and the corresponding canonical incomplete preference patterns.}
\label{fig:gadget-literals}
\end{figure}

\smallskip

Define the \emph{positive gadget} to be the graph $H$ constructed as in Figure~\ref{fig:gadget-literals}, left. Node $a$ is called the \emph{apex} of $H$. Node $c$ is called the \emph{consistency enforcer} of $H$. Node $g$ is called the \emph{gateway} of $H$. Note that Figure~\ref{fig:gadget-literals}, left, also defines an incomplete preference pattern $I_H$, which we call \emph{canonical}.

Define the \emph{negative gadget} to be the graph $N$ constructed as in Figure~\ref{fig:gadget-literals}, right. Node $a$ is called the \emph{apex} of $N$. Node $c$ is called the \emph{consistency enforcer} of $N$. Node $g$ is called the \emph{gateway} of $N$. Edge $eh$ is called the \emph{evicted} edge of $N$.
Figure~\ref{fig:gadget-literals}, right, also defines an incomplete preference pattern $I_N$, which we call \emph{canonical}.
	
To each positive (resp. negative) literal $\ell$ of a restricted $3-$SAT instance, we associate one positive (resp. negative) gadget $H(\ell)$ (resp. $N(\ell)$). We will write $G(\ell)=(V(\ell),E(\ell))$ to denote the gadget associated to a literal $\ell$, whether it is positive or negative. Nodes of $G(\ell)$ will be denoted by $a(\ell)$, $b(\ell)$, ... .

\smallskip

Consider a $3$-SAT instance $\psi$. Fix an arbitrary order of clauses and of literals in each clause. Define the graph $G:=G(\psi)$ as follows. First, construct disjoint graphs $G(\ell)$ for each literal $\ell$, and add nodes $u$ and $v$. Now fix a clause $c$, and let $\ell_1,\ell_2, \ell_3$ be the literals from the clause, in this order (with $\ell_3$ possibly not present). Identify the apex of $G(\ell_1)$ with $u$; the apex of $G(\ell_i)$, $i\geq 2$ with the gateway of $G(\ell_{i-1})$; and add edge $g(\ell_k)v$, where $\ell_k$ is the last literal of the clause.  Repeat the operation for all clauses (note that vertices $u$ and $v$ are unique). Let $\bar G(\bar V,\bar E)$ be the graph obtained, and define
$$V(G):=\{s,t,w,x,y\}\cup \bar V.$$

Define $E(G):= \bar E \cup E_1 \cup E_2$, where:
\begin{itemize}
	\item[]$E_1:=\{st,tu,vw,wx,xy\}$;
%	\item[]$E_2:=\cup_{k \in {\cal K}}\{g(\ell)v : \hbox{ $\ell$ is the last literal of clause $k$}\}$;
	\item[]$E_2:=\cup_{i=1,...,n} \{c(\ell)c(\ell') : \hbox{ $\ell$ is a positive occurrence of variable $i$} \\  \quad \hbox{ and $\ell'$ is a negative occurrence of variable $i$}\} \quad $ (\emph{consistency edges}). %\quad$ (``entrance edges''), 
%	\item[]$E_3:=\cup_{i=1,2;k \in {\cal K}} \{g(\ell)a(\ell') : \hbox{ $\ell$ is the $i$-th and $\ell'$ is the $(i+1)$-th literal of clause $k$}\} \quad$ (``next edges''),
%	\item[]$E_4:=\cup_{k \in {\cal K}} \{g(\ell)t : \hbox{ $\ell$ is the last literal of clause $k$}\}\quad$ (``exit edges''),
%
\end{itemize}
\noindent See Figure~\ref{fig:example} for an example of the graph $G(\psi)$.

\begin{lemma}
Let $G(\psi)$ be defined as above for a monotone $3-SAT$ instance $\psi$. Then $G(\psi)$ is bipartite.
\end{lemma}
\begin{proof}
We assign each vertex from $V:=V(G(\psi))$ to exactly one of two sets $A$ or $B$, and prove that if $u,v \in A$ (or $u,v \in B$) then $uv \notin E:=E(G(\psi))$. The proof can be followed in Figure \ref{fig:example}, where nodes are colored according to the set of the bipartition they belong to. Assign $s, u, w, y$ to $A$ and $t, v, x$ to $B$. For any $\ell$ corresponding to a positive literal, assign $a(\ell), d(\ell), c(\ell), g(\ell)$ to $A$ and $b(\ell), e(\ell), f(\ell)$ to $B$. For any $\ell$ corresponding to a negative literal, assign $a(\ell), d(\ell), e(\ell), g(\ell)$ to $A$ and $b(\ell), c(\ell), f(\ell), h(\ell)$ to $B$. Note that this assignment is consistent with the fact that some $a(\ell)$ and $g(\ell)$ are identified, and that some $a(\ell)$ are identified with $u$ (they are all assigned to $A$). 
One easily checks that edges in $\bar E \cup E_1$ connect nodes in the same set of the bipartition. Now consider $c(\ell)c(\ell') \in E_2$. By construction, we can assume wlog that $\ell$ (resp. $\ell'$) is a positive (resp. negative) literal. Hence $c(\ell) \in A$, $c(\ell') \in B$, as required. \end{proof}

%Consider consistency edge $c(\ell)c(\ell')$. We conclude that the above bipartition is valid by noticing that $c(l) \in A$ and $c(l) \in B$, and that there is no $uv \in E$ such that $u,v \in A$ or $u,v \in B$.

Our hardness results are based on the following lemma.

\begin{lemma}
\label{lem: 3SAT-popular}
Consider the graph $G:=G(\psi)$, and set $F:= \{st,wx\}.$
There exists a preference system $<$ with the following property: $(G(\psi),<)$ has a popular matching that does not contain any edge in $F$ if and only if $\psi$ is satisfiable. Given $\psi$, this preference system can be found in polynomial time.
\end{lemma}

We devote the rest of Section~\ref{sec:reduction} to the proof of Lemma~\ref{lem: 3SAT-popular}. The reader interested in how Lemma~\ref{lem: 3SAT-popular} is applied to deduce the proofs of main results can skip directly to Section~\ref{sec:main-results}.

\begin{figure}[h]
\centering
\includegraphics[scale = 0.85]{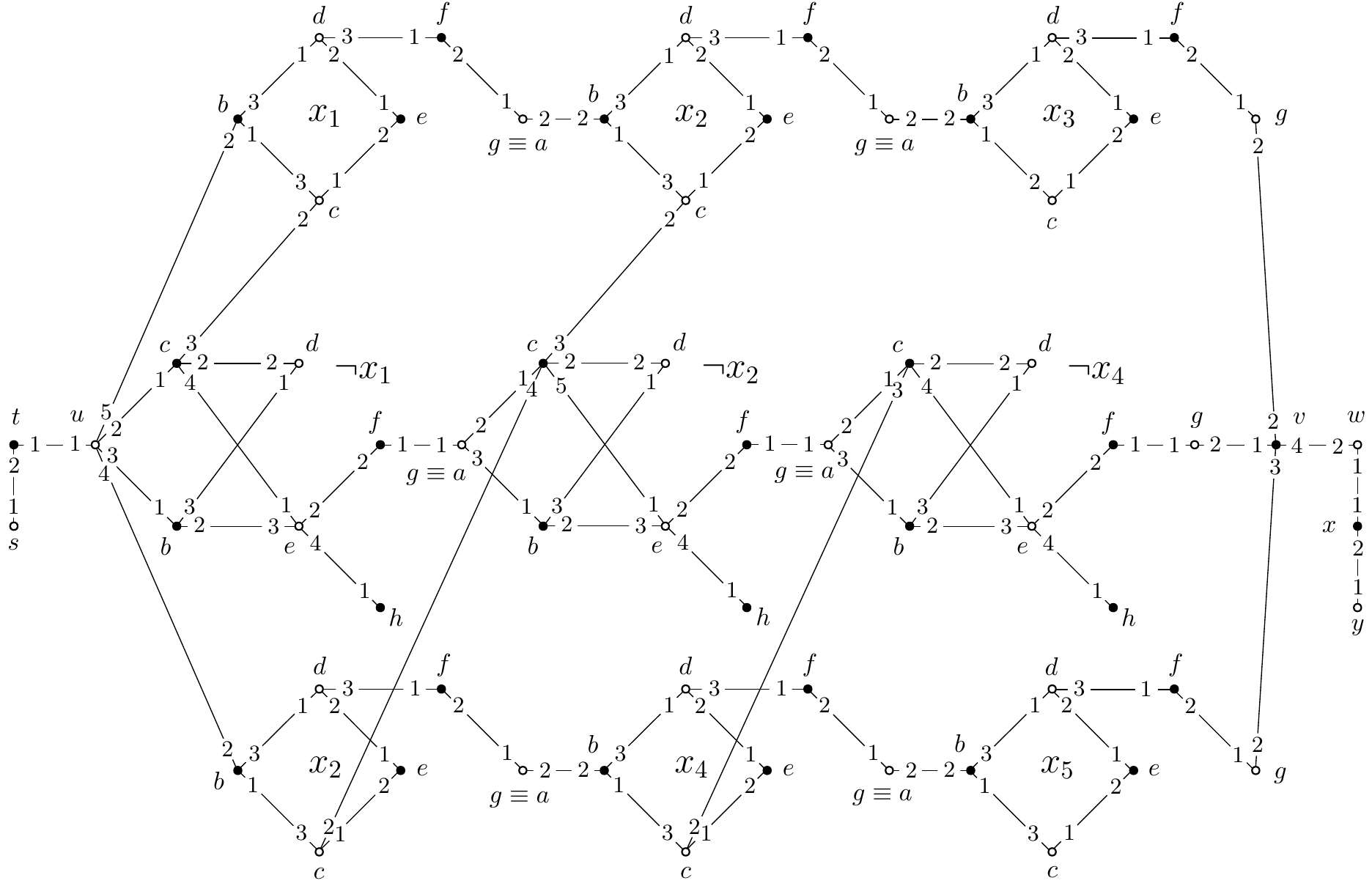}
\caption{A monotone $3$-SAT instance $\psi = (x_1 \lor x_2 \lor x_3) \land (\neg x_1 \lor \neg x_2 \lor \neg x_4) \land (x_2 \lor x_4 \lor x_5)$  and the associated graph $G(\psi)$ and preference system. Disk and circle labels of vertices provide a bipartition of $G(\psi)$.}
\label{fig:example}
\end{figure}

\subsection{The preference system}

In the remaining of Section~\ref{sec:reduction}, we fix a monotone 3-SAT instance $\psi$ and the corresponding graph $G(V,E):=G(\psi)$. We start by showing how to define the preference system via the preference pattern $L$.

Consider the union $I$ of all canonical incomplete preference patterns for gadgets corresponding to literals of $\psi$. Recall that vertex sets of gadgets are disjoint subsets of $V$, with the exception of $g(\ell_1)\equiv a(\ell_2)$, with $\ell_2$ being the literal following $\ell_1$ in some clause, and $u \equiv a(\ell)$, with $\ell$ first literal in some clause. However, the canonical incomplete preference patterns assign $\emptyset$ to edges incident to nodes $a(\ell)$. Hence, $I$ is an incomplete preference pattern for $G$. 

Consider the preference system $L$ that: agrees with $I$; has values $L(s,t)=1$, $L(t,s)=2$, $L(u,t)=L(t,u)=1$, $L(v,w)=|\Gamma(v)|$, $L(w,v)=2$, $L(w,x)=1$, $L(x,w)=1$, $L(x,y)=2$, $L(y,x)=1$; assigns $L$ for the remaining ordered pairs of adjacent nodes for which it has not been defined yet.

%If $(u,v)$ is an edge in $G$ and $L(u,v)$ has not been assigned, assign $L(u,v)$ such that 
%
%\begin{enumerate}
%	\item $L(u,v)<|\Gamma(u,v)|$
%	\item $L(u,v)\neq L(u,w)$ when $(u,w)\in E$
%\end{enumerate}
%
%Repeat until all pairs $(u,v)$ such that $(u,v)\in E$ have preferences.\\

\begin{itemize}
\item[-]
$L(u,b(\ell))$, for all literals $\ell$ that are first in some clause, and $uc(\ell)$ for all negative literals that are first in some clause. Since $L(u,t)=1$ and no other edge is incident to $u$, we can assign $L(u,v)$ for $v=b(\ell)$ or $v \in \{b(\ell), c(\ell)\}$ as above, through an arbitrary bijection to $\{2,\dots, |\Gamma(v)|\}$.%For positive literals $\ell$, we have $\Gamma(b(\ell))=\{u,c(\ell), d(\ell)\}$ and we already set $L(b(\ell), c(\ell))=1$, $L(b(\ell), d(\ell))=3$. We therefore set $L(b(\ell), u)=2$. For $\ell$ negative literals, with an argument similar to the one above, we set $L(b(\ell), u))=1$. Since $L(c(\ell), d(\ell))=2$ and $L(c(\ell), e(\ell))=|\Gamma(c)|\geq3$, we set moreover $L(c(\ell), u)=1$.
\item[-]
$L(g(\ell),v)$ and $L(v,g(\ell))$, for all literals $\ell$ that are last in some clause. Since all such $g(\ell)$ have degree 2 in $G$, and $L(g(\ell), f(\ell))=1$, we set $L(g(\ell), v)=2$. Moreover, $\Gamma(v)$ is composed of all such $g(\ell)$ and vertex $w$, to which we already assigned $L(v,w)=|\Gamma(v)|$. We assign therefore $L(v,g(\ell))$ via an arbitrary bijective map to $\{1,\dots,|\Gamma(v)|-1\}$.
\item[-]
$L(c(\ell),c(\ell'))$ and $L(c(\ell'),c(\ell))$, where $\ell$ is a positive occurrence and $\ell'$ is a negative occurrence of the same variable. Note that we already assigned $L(c(\ell), e(\ell))=1$, $L(c(\ell), b(\ell))=|\Gamma(c(\ell))|$. We assign $L(c(\ell), c(\ell'))$ for a fixed $c(\ell)$ via an arbitrary bijective map to $\{2,\dots,|\Gamma(c(\ell))|-1\}$. Similarly, we assign $L(c(\ell'), c(\ell))$ for a fixed $c(\ell')$ via an arbitrary bijective map to $\{3,\dots,|\Gamma(c(\ell'))|-1\}$
\item[-]
$L(a(\ell),b(\ell))$ for all positive literals $\ell$ that are not first in some clause. Recall that we have $L(a(\ell),f(\ell'))=1$ for some $\ell'$. We set therefore $L(a(\ell),b(\ell))=2$.
\item[-]
$L(a(\ell),b(\ell))$, $L(a(\ell),b(\ell))$ for all negative literals $\ell$ that are not first in some clause. Recall that we have $L(a(\ell),f(\ell'))=1$ for some $\ell'$. We set therefore $L(a(\ell),b(\ell))=3$, $L(a(\ell),c(\ell))=2$.

\end{itemize}

We let $(G,<)$ be the preference system corresponding to the preference pattern $L$. We now start deriving properties of matchings that are popular in $(G,<)$ and satisfy some further hypothesis. In the proof of the next and following lemmas, we often deduce a contradiction using the characterization of popular matchings given in Theorem~\ref{thr:characterize-popular}, hence we omit referring explicitly to this theorem each time.

\subsection{Properties of certain popular matchings of $(G,<)$}

\begin{lemma}
\label{no consistency edges}
Let $M$ be a popular matching of $(G,<)$ that does not contain any evicted edge. Then $M$ does not contain any consistency edge, i.e., $M\cap E_2=\emptyset$. \end{lemma}
\begin{proof}
	Suppose by contradiction that $M$ contains a consistency edge, connecting positive gadget $H(\ell)$ and negative gadget $N(\ell')$. We focus now on nodes of $N(\ell')$, hence omit the dependency on $\ell'$. Since the consistency edge incident to $c$ is matched, if $d$ is not matched, then $cd$ would be a $(+,+)$ edge adjacent to an unmatched vertex, which contradicts popularity. Hence, $bd\in M$, which in turn implies $ef\in M$, otherwise $be$ is a $(+,+)$ edge adjacent to an unmatched vertex, again a contradiction (recall that $eh \notin M$ by hypothesis). This implies that $he$ is a $(+,-)$ edge, $fg$ is a $(+,+)$ edge, and $h,e,f,g$ is an $M$-alternating path containing a $(+,+)$ edge, a contradiction. \end{proof}

\begin{lemma}
\label{M contains (t,u), apex is matched to its first choice}
	Let $M$ be a popular matching of $(G,<)$ that contains $tu$. Then:
	
	\begin{enumerate}
	\item\label{uno} The gateway (resp. apex) of each gadget is matched to its favorite partner;
	\item\label{due} $M$ does not contain any consistency or evicted edge;
	\item\label{tre} In each positive gadget $H(\ell)$, exactly one of the following is true (dependency on $\ell$ is omitted):
	
	\begin{enumerate}
		\item $\mathcal{F}(H):=\{bd, ce\}\subseteq M$, or 
		\item $\mathcal{T}(H):=\{bc, de\}\subseteq M$.
	\end{enumerate}
	
	Similarly, in each negative gadget $N(\ell')$, exactly one of the following is true (dependency on $\ell'$ is omitted):
	
	\begin{enumerate}
		\item $\mathcal{F}(N):=\{cd, be\}\subseteq M$, or 
		\item $\mathcal{T}(N):=\{ce, bd\}\subseteq M$.
	\end{enumerate}
	
	\item\label{quattro} If $\ell$ is a positive and $\ell'$ a negative occurrence of the same variable, then we cannot have both $\mathcal{T}(H(\ell))\subseteq M$ and $\mathcal{T}(N(\ell))\subseteq M$.

	\end{enumerate}
\end{lemma}
\begin{proof}

\ref{uno}. We consider negative and positive gadgets separately. Note that it is enough to prove the statement for the gateway, since the one for the apex follows by the fact that each apex is either also a gateway, or $u$ (that is matched to its favorite partner $t$ by hypothesis).

Pick first a negative clause and one of its literals $\ell$. We omit the dependency on $\ell$ in nodes of $N(\ell)$. The favorite partner of $g$ is $f$. Hence, suppose by contradictiong $fg \notin M$. Note that since $g$ is also $f$'s favorite partner, $fg$ is a $(+,+)$ edge, which implies $fe \in M$, else $f$ is an unmatched node adjacent to a $(+,+)$ edge starting from an unmatched node, a contradiction. Edge $eh$ is therefore labeled $(-,+)$. The path $h,e,f,g$ is therefore an $M$-alternating path to a $(+,+)$ edge, a contradiction concluding the proof for this case.

Consider any positive clause, and take, among its literals, the first one that violates the thesis. Call this literal $\ell$. Again, we omit the dependency on $\ell$ in nodes of $H(\ell)$, and observe that $f$ is $g$'s favorite partner. Therefore, $fd \in M$, else $f$ is unmatched and $fg$ is a $(+,+)$ edge adjacent to an unmatched node, a contradiction. If $e$ is unmatched, then $de$ is a $(+,+)$ adjacent to an unmatched node, again a contradiction. Since $\Gamma(e)=\{d,c\}$ and $df \in M$, we conclude $ec \in M$. A similar argument (with $e$ replaced by $b$) implies $ab \in M$. Now, if $\ell$ is the first literal of the clause, $a\equiv u$, contradicting $tu \in M$. Else, $a\equiv g(\ell')$ for the literal $\ell'$ preceding $\ell$ in the clause, contradicting the choice of $\ell$.

\smallskip

\ref{due}. 	Suppose by contradiction $M$ contains an evicted edge. Pick any clause with a negative literal that contains an evicted edge, and let $N(\ell)$ be the gadget associated to the first such literal. We omit the dependency on $\ell$ in nodes of $N(\ell)$. Part~\ref{uno} implies that $a$ is matched to its favorite partner. The hypothesis implies that $eh \in M$. Hence, if $bd\notin M$, then $be$ is a $(+,+)$ edge incident to the unmatched node $b$, a contradiction. Hence $bd \in M$. Since $a,d,e$ are not matched to $c$ in $M$, if $c$ is not matched through a consistency edge it is unmatched, and $ce$ is $(+,+)$ edge, a contradiction.

Hence $cc(\ell') \in M$ for some $H(\ell')$. Arguments similar to those above imply $d(\ell')e(\ell') \in M$ and $a(\ell')b(\ell') \in M$, contradicting part~\ref{uno}. This concludes the proof that $M$ does not contain any evicted edge. The fact that $M$ does not contain any consistency edge immediately follows from what just proved and Lemma~\ref{no consistency edges}.

\smallskip

\ref{tre}. Fix a gadget $G(\ell)$ and omit dependency on $\ell$ in nodes. By part~\ref{uno} and~\ref{due}, $a$ is matched to its first choice, and $M$ does not contain any consistency or evicted edges. Hence, we can restrict our attention to the subgraph induced by $b,c,d,e$. Simple case checking shows that, if any of those is unmatched, then there is a $(+,+)$ edge incident to an unmatched node, a contradiction to popularity. Hence, each of these vertices must be matched, and since they form a cycle with $4$ vertices, one of the possibilities from the thesis of the lemma must hold.

\smallskip

\ref{quattro}. 	Suppose we have both $\mathcal{T}(H(\ell))\subseteq M$ and $\mathcal{T}(N(\ell'))\subseteq M$. Then the consistency edge $c(\ell)c(\ell')$ is $(+,+)$. Hence, $h(\ell'),e(\ell'),c(\ell'),c(\ell)$ is an $M$-alternating path from an unmatched node to a $(+,+)$ edge, a contradiction. 

\end{proof}

\iffalse

\begin{lemma}
\label{no undesirable edges}
	Let $M$ be a popular matching of $G(\psi,<)$ which contains $tu$. Then s.
\end{lemma}

\begin{proof}

\end{proof}

\begin{corollary}
\label{no consistency edges in popular matching}
	If $M$ is a popular matching of $G(\psi,<)$ which contains $tu$, then it does not contain any consistency edges.
\end{corollary}

\begin{proof}

\end{proof}

\begin{lemma}
\label{assigning false or true}
	If $M$ is a popular matching of $G(\psi,<)$ which contains $tu$, then 
	
\end{lemma}

\begin{proof}
	\end{proof}

We now show that popular matching in a positive and a negative gadget that refer to the same variable are related. 

%label the sets in Lemma~\ref{assigning false or true} as $\mathcal{F}$ and $\mathcal{T}$ to provide an association to truth values of the corresponding literal in $3-SAT$ instance. We make this precise in the next section.

\begin{lemma}
\label{true and false are not possible simultaneously}
	Let $M$ be a popular matching of $G(\psi,<)$ and whose gadgets If $H(\ell)$ and $N(\ell')$ are connected by a consistency edge. Then we cannot have both $\mathcal{T}(H)\subseteq M$ and $\mathcal{T}(N)\subseteq M$.
\end{lemma}

\begin{proof}

\end{proof}

\fi

\subsection{From popular matchings to feasible assignments}

Let $M$ be a popular matching $M$ of $(G,<)$ that does not contain edges in $F=\{st,wx\}$. Usual arguments imply $tu, xy, vw \in M$. Observe that $wx$ is a $(+,+)$ edge and $s$ is unmatched.

% and there cannot be an alternating path from $s$ to $w$, else Theorem~\ref{thr:characterize-popular} is not satisfied.\\
	
	We claim that each clause contains at least one literal $\ell$ with $T(\ell)\subseteq M$. We show that, if this does not happen, then there exists an $M$-alternating path from $s$ to $x$ passing through $wx$ (in the following, we say that such a path \emph{goes from $s$ to $wx$}). Together with the fact that $wx$ is a $(+,+)$ edge and $s$ is unmatched, this contradicts the popularity of $M$. The construction of those paths can be followed in Figure~\ref{fig:alt path positive and negative clauses}. We assume that the clause is given by exactly three positive literals. Similar arguments apply for the case with less than three positive literals. Recall that, by Lemma~\ref{lem: 3SAT-popular}, part 1, apexes and gateways are matched to their favourite partner.
	
	Suppose first we deal with a positive clause, and let $H(\ell_1), H(\ell_2), H(\ell_3)$ be the corresponding gadgets. By Lemma~\ref{M contains (t,u), apex is matched to its first choice}, part 3, $\mathcal{F}(\ell_1), \mathcal{F}(\ell_2), \mathcal{F}(\ell_3) \subseteq M$. Suppose first $H(\ell)$ is a positive gadget. Then $st$, $a(\ell)b(\ell)$, $d(\ell)f(\ell)$, $g(\ell_3)v$ are $(+,-)$ or $(-,+)$ edges for $\ell=\ell_1,\ell_2,\ell_3$, and the promised $M$-alternating path is:
$$s,t,u,b(\ell_1),d(\ell_1),f(\ell_1),g(\ell_1),b(\ell_2),d(\ell_2),
f(\ell_2),g(\ell_2),b(\ell_3),d(\ell_3),f(\ell_3),g(\ell_3),v,w,x$$
Consider now a clause that is given by three negative literals: $N(\ell_1), N(\ell_2), N(\ell_3)$. Then $st$, $a(\ell)b(\ell)$, $e(\ell)f(\ell)$, $g(\ell_3)v$ are $(+,-)$ or $(-,+)$ edges for $\ell=\ell_1,\ell_2,\ell_3$, and the the alternating path is $$s,t,u,b(\ell_1),e(\ell_1),f(\ell_1),g(\ell_1),b(\ell_2),e(\ell_2),
f(\ell_2),g(\ell_2),b(\ell_3),e(\ell_3),f(\ell_3),g(\ell_3),v,w,x,y.$$

\begin{figure}
\begin{center}
\includegraphics[scale=0.85]{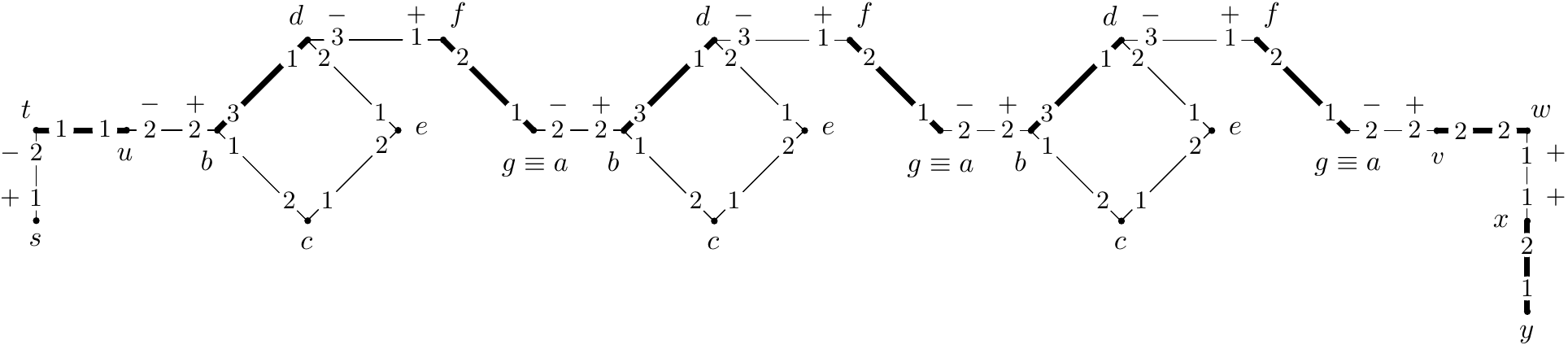}\hspace{1cm}
\includegraphics[scale=0.85]{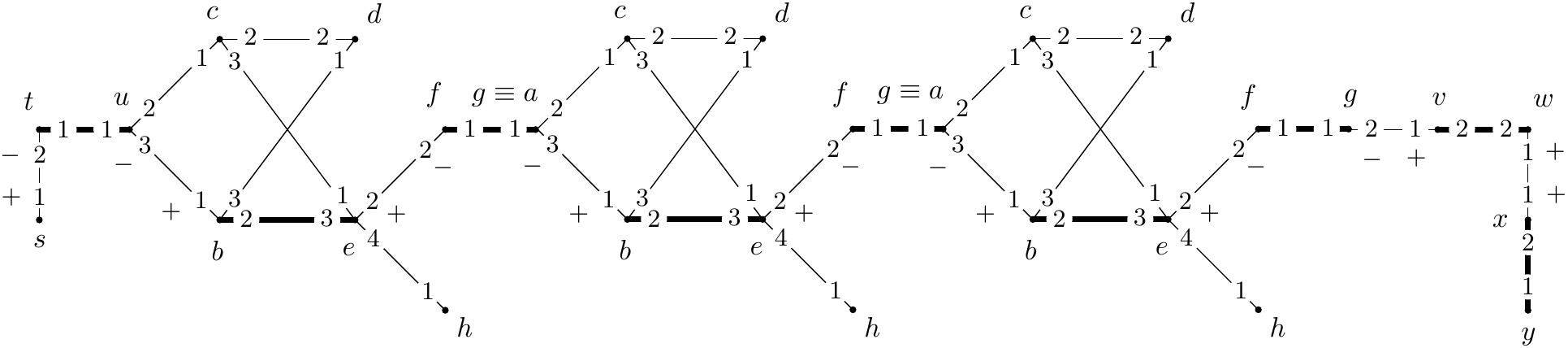}
\end{center}
\caption{An alternating path from $s$ to $w$ when all literals are False in positive (above) and negative (below) clauses.}
\label{fig:alt path positive and negative clauses}
\end{figure}
	
From $M$, we obtain an assignment of the literals for $3-SAT$ as follows: if $\mathcal{T}(\ell) \subseteq M$ for some positive (resp. negative) occurrence $\ell$ of variable $x_i$, then set $x_i$ to be true (resp. false). Else, set $x_i$ to true or false arbitrarily. Observe first that the assignment is well-defined, since we cannot assign the same variable to simultaneously true and false, because of Lemma~\ref{M contains (t,u), apex is matched to its first choice}, part~\ref{quattro}. Since by the first part of the proof, in each clause there is at least an $\ell$ such that $T(\ell)\subseteq M$, we conclude that $\psi$ is satisfied by the assignment we constructed.	

\subsection{From feasible assignments to popular matchings}

	Now suppose that we have an assignment of the literals which satisfies $\psi$. Construct a matching $M$ as follows. First, add edges $tu, vw, xy$, and edges $f(\ell)g(\ell)$ for each literal $\ell$. Recall that $f(\ell)$ is the favourite partner of $g(\ell)$. Then, for each positive literal $\ell$, add $\mathcal{T}(\ell)$ to $M$ if the corresponding variable is set to true; add $\mathcal{F}(\ell)$ to $M$ if the corresponding variable is set to false. Last, for each negative literal $\ell$, add $\mathcal{T}(\ell)$ to $M$ if the corresponding variable is set to false; add $\mathcal{F}(\ell)$ to $M$ if the corresponding variable is set to true. Figure~\ref{fig:assignment} gives an example of a matching $M$ corresponding to a feasible assignment of the clause from Figure~\ref{fig:example}. Note that $M$ does not contain $st$ and $wx$, hence it does not contain edges in $F$. We will show that $M$ is popular by proving it satisfies the conditions of Theorem~\ref{thr:characterize-popular}. This concludes the proof.
	
\begin{figure}
\begin{center}
\includegraphics[scale=0.85]{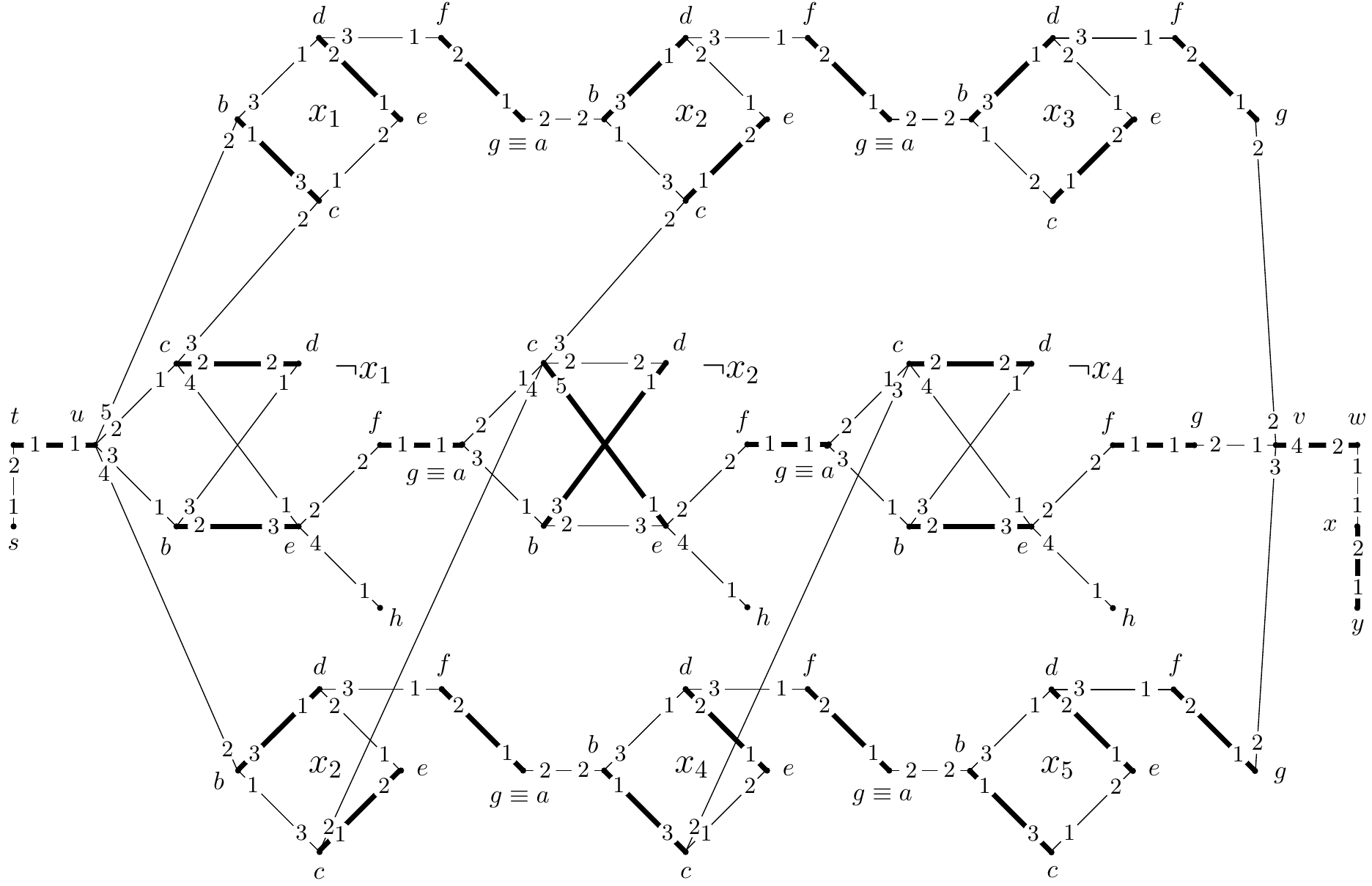}
\end{center}
\caption{A matching $M$ (bold edges) corresponding to a feasible assignment of $\psi = (x_1 \lor x_2 \lor x_3) \land (\neg x_1 \lor \neg x_2 \lor \neg x_4) \land (x_2 \lor x_4 \lor x_5)$ from Figure~\ref{fig:example}: $x_1,x_4,x_5$ are set to True and $x_2,x_3$ to False.}
\label{fig:assignment}
\end{figure}	

Define an $M$-alternating path $P$ to be \emph{malicious} if it starts from an unmatched node and contains a $(+,+)$ edge. Observe that, by construction, the only unmatched nodes are $s$, $h(\ell)$ for all negative literals $\ell$. The following claim caracterizes the only $(+,+)$ edge.

\begin{claim} \label{claim: no consistency edge is (+,+)}
The only edge labelled $(+,+)$ is $wx$.
\end{claim}
\begin{proof}
Consider first a consistency edge $c(\ell)c(\ell')$, which connects a positive gadget $H(\ell)$ and a negative gadget $N(\ell')$. By construction, one of the following two cases happen: ${\cal T}(H(\ell))$, ${\cal F}(H(\ell')\subseteq M$, or ${\cal T}(H(\ell'))$, ${\cal F}(H(\ell)\subseteq M$. One easily checks that $c(\ell)c(\ell') = (-,+)$ or $c(\ell)c(\ell') = (+,-)$. 

Notice that $st = (+,-)$, and all edges of $E(G)\setminus M$ incident to an apex or a gateway have at least a label equal to $-$ since, by construction, those nodes are matched to their favourite partners. Hence, we are left with edges between nodes $b,d,c,e,f$ and possibly $h$ of the same gadget $G(\ell)$ (we omit dependency on $\ell$ in nodes). Suppose first $\ell$ is a positive literal. If $\ell$ is true, then $b,e$ are matched to their favorite partners. If $\ell$ is false, $c,d$ are matched to their favorite partners. Hence, no 
$(+,+)$ edge is incident to those. In all cases, $df$ is a $(-,+)$ edge. Suppose now $\ell$ is a negative literal. If $\ell$ is true, then $e,c$ are matched to their favorite partners, concluding the proof. Else, $cd, be \in M$, and then $db,ce,eh,fe$ are $(-,+)$ edges.\end{proof}

 %Consider a positive gadget $H(\ell)$. $ub(\ell) = (--)$ if ${\cal T}(H(\ell))\subseteq M$, or $ub(\ell) = (-+)$ if ${\cal F}(H(\ell))\subseteq M$, ${g(\ell)v = (-+)}$.
%If ${\cal T}(H(\ell))\subseteq M$, then $b(\ell)d(\ell)=(-,+)$, ${c(\ell)e(\ell)=(+,-)}$, ${a(\ell)b(\ell)=(-,-)}$. If ${\cal F}(H(\ell))\subseteq M$, then $b(\ell)c(\ell)=(+,-)$, $e(\ell)d(\ell)=(+,-)$, $a(\ell)b(\ell)=(-,+)$. $d(\ell)f(\ell)=(-,+)$ in both cases. Consider a negative gadget $N(\ell')$. ${ub(\ell') = (-+)}$, ${uc(\ell') = (-+)}$ in both cases, ${g(\ell')v = (-+)}$.
%If ${\cal T}(N(\ell'))\subseteq M$, then $c(\ell')d(\ell')=(+,-),b(\ell')e(\ell')=(+,-),e(\ell')f(\ell')=(-,-)$. If ${\cal F}(N(\ell'))\subseteq M$, then $c(\ell)e(\ell)=(-,+),b(\ell)d(\ell)=(-,+),e(\ell)f(\ell)=(+,-)$. $d(\ell)f(\ell)=(-,+)$ in both cases.

Claim~\ref{claim: no consistency edge is (+,+)} implies that conditions $(i)$ and $(ii)$ of Theorem~\ref{thr:characterize-popular} are satisfied.

\begin{claim}\label{cl:each clause contains true literal}
If a positive clause contains a true literal $\ell$, then $a(\ell),b(\ell)=(-,-)$. If a negative clause contains a true literal $\ell'$, then $e(\ell')f(\ell')=(-,-)$.
\end{claim}
\begin{proof}
Consider $H(\ell)$ and $N(\ell')$. By hypothesis,  $b(\ell)c(\ell), d(\ell)e(\ell) \in M$ and $c(\ell')e(\ell'),\break b(\ell')d(\ell') \in M$. This implies that $a(\ell)b(\ell)=(-,-)$ and $e(\ell')f(\ell')=(-,-)$.
\end{proof}

\begin{claim} \label{claim: no malicious path in a clause}
For any clause $c$, there is no malicious path from $s$ to $wx$ that is contained in $\{s,t,u,v,w,x\}\cup \bigcup_{\ell \in c}V(\ell)$.
\end{claim}
\begin{proof}
Since we have an assignment of the literals which satisfies $\psi$, then each clause contains at least one true literal. According to Claim~\ref{cl:each clause contains true literal}, $a(\ell)b(\ell)=(-,-)$ and $e(\ell')f(\ell')=(-,-)$ for some literals $\ell$ and $\ell'$ in positive and negative clauses correspondingly. Since any path that satisfies the hypothesis with $c$ a positive (resp. negative) clause must pass through edges $a(\ell)b(\ell)$ (resp. $e(\ell)f(\ell)$) for all $\ell \in c$, the thesis follows. 
\end{proof}

\begin{claim} \label{no malicious path from h in a clause}
A malicious path starting at $h(\ell)$ does not traverse $f(\ell)$. In particular, there is no malicious path from $h(\ell)$ to $wx$ that is contained in $\{s,t,u,v,w,x\}\cup \bigcup_{\ell \in c}V(\ell)$ for any negative clause~$c$.
\end{claim}
\begin{proof}
$h(\ell)$ is unmatched and $e(\ell)f(\ell) \notin M$ regardless of the true or false assignment to literal~$\ell$. Hence $P$ cannot take both $h(\ell)e(\ell)$ and $e(\ell)f(\ell)$, since none of them is in $M$. Since $f(\ell)$ has degree $2$, $f(\ell)\notin P$. We conclude that there is no alternating path from $h(\ell)$ to $wx$ contained only in $c$.
\end{proof}

\begin{claim} \label{in positive literal f not in P}
Let $\ell$ be a positive literal, and let $P$ be a malicious path. Then $c(\ell)c(\ell')$ is an edge of $P$ for at most one $\ell'$ and if that happens $f(\ell) \notin P$.
\end{claim}
\begin{proof}
The first statement follows from the definition of path. Now assume there is $\ell'$ such that $c(\ell)c(\ell') \in P$. We omit dependency on $\ell$ in nodes. By construction, $cc(\ell')\notin M$. Suppose by contradiction $f \in P$. Since $f$ has degree $2$, $fd$ is an edge of $P$. We have two possibilities. First, if $\ell$ is true, then $bc, de \in M$. Then $P$ or its inverse (i.e.~the path obtained traversing nodes from $P$ in opposite order) contains the subpath $f,d,e,c,c(\ell')$, a contradiction to the fact that $P$ is $M$-alternating. Second, if $\ell$ is false, then $bd, ce \in M$. Then $P$ or its inverse contains the subpath $f,d,b,c,c(\ell')$, again a contradiction to the fact that it is $M$-alternating. \end{proof}

\begin{claim} \label{claim: a(ell) is not in P}
Let $\ell$ be a negative literal, and let $P$ be a malicious path. Then $c(\ell),c(\ell')$ is an edge of $P$ for at most one $\ell'$ and if that happens and $a(\ell)\neq u$, then $a(\ell) \notin P$.
\end{claim}
\begin{proof}
Assume there is $\ell'$ such that $c(\ell)c(\ell') \in P$. Again, we omit explicit dependency on $\ell$ in nodes, and the first part of the statement is immediate. If $a\neq u$, then $\Gamma(a)=\{f(\ell''),b,c\}$ for some $\ell''$. We have two possibilities. First, if $\ell$ is true, then $ce, bd \in M$. This implies that $P$ or its inverse contains the subpath $c(\ell'),c,e$. Hence if $a \in P$, then $ac$ is not an edge of $P$. Since $af(\ell'') \in M$ for some $\ell''$ by construction and $P$ is $M$-alternating, both $ab$ and $bd \in M$ are edges of $P$. But then we must have that $cd$ is an edge of $P$, a contradiction to $P$ being a path. Second, if $\ell$ is false, then $cd, be \in M$. Hence, either $P$ or its inverse contains the subpath $c(\ell')c,d,b,e$. Thus, $ac, ab$ are not edges of $P$, which implies that $a \notin P$.
\end{proof}

In order to conclude that condition $(iii)$ of Theorem~\ref{thr:characterize-popular} is satisfied, suppose, by contradiction, there exists a malicious path $P$ starting from $s$. Then $P=s,t,u,\dots$. Claim~\ref{claim: no malicious path in a clause} implies that there is at least one consistency edge that belongs to $P$. Take the first consistency edge $c(\ell)c(\ell')$ that $P$ traverses, and suppose that $c(\ell)$ is traversed by $P$ before $c(\ell')$. This means that $a(\ell)$ and $z$ are two consecutive nodes of $P$, traversed in this order and before $c(\ell)$, for some $z \in V(\ell)$ (possibly $z=c(\ell)$). Suppose first $\ell$ is a positive literal. Hence, $a(\ell),b(\ell) \in P$. By Claim~\ref{cl:each clause contains true literal}, $\ell$ is false, hence $\ell'$ is true. Using again Claim~\ref{cl:each clause contains true literal}, we deduce $f(\ell')\notin P$. Hence $P$, in order to reach $wx$, must traverse $a(\ell')$ after $c(\ell')$. If $a(\ell')=u$, then $P$ cannot traverse $u$ again. Hence $a(\ell')\neq u$ and, by Claim~\ref{claim: a(ell) is not in P}, $a(\ell') \notin P$. In both cases, $P$ cannot reach $wx$, a contradiction. If conversely $\ell$ is a negative literal, by Claim~\ref{claim: a(ell) is not in P} $a=u(\ell)$. If $\ell$ is set to true, then by construction $cd, be \in M$. Hence, $P$ contains the subpath $u,b,c,d,e$ or the subpath $u,c,e$ (we omitted dependency on $\ell$). In both cases, $c(\ell)c(\ell')$ cannot be an edge of $P$, a contradiction. If conversely $\ell$ is set to false, then $P$ contains the subpath $u,b,e$ or $u,c,d$ (again, we omitted the dependency on $\ell$). In the latter case, $cc(\ell')$ cannot be an edge of $P$. In the former, the inverse of $P$ must contain the subpath $c(\ell'),c,d,b$, a contradiction ($b$ would have degree $3$ in $P$).

%True then $a(\ell')b(\ell')$ is a $(-,-)$ edge, and $P$ is not malicious. If $x$ is False, then $\mathcal{T}(\ell),\mathcal{F}(\ell')\subseteq M$. $P$ exits $N(\ell)$ via $e(\ell)f(\ell)$ by Claim~\ref{claim: a(ell) is not in P}, but $e(\ell)f(\ell)= (-,-)$, contradiction to $P$ being malicious.

Suppose now there exists a malicious path $P$ starting from $h(\ell'')$. Since $tu \in M$, $u \notin P$. From Claim~\ref{no malicious path from h in a clause}, $P$ must traverse (in this order) a consistency edge $c(\ell')c(\ell)$ with $\ell'$ being a false literal. Pick the first such edge. By Claim~\ref{in positive literal f not in P}, $f(\ell) \notin P$. Hence, after traversing $c(\ell)$, $P$ visits only gadgets corresponding to literals preceding $\ell$ in the clause, until it traverses (in this order) another consistency edge $c(\ell''')c(\ell'''')$, since as we argued above, $u \notin P$. This contradicts Claim~\ref{in positive literal f not in P}, since it must also be that $f(\ell''') \in P$.

% But this implies Hence, $P$ traverses $f(\ell'')$ for some positive literal $\ell''$ that precedes $\ell$ in the clause. 
% a gadget $H(\ell'')$ for some positive literal $c(\ell'')c(\ell''')$ and $\ell$ there exists a negative literal $\ell'$ preceeding $c$ and a positive literal $\ell$ such that there is an alternating subpath $P'$ of $P$ contained in vertices from gadgets in $c$ between $h(\ell'')$ and $h(\ell')$ and the consistency edge $c(\ell')c(\ell) \in P$. By Claim~\ref{cl:each clause contains true literal}, $\ell'$ is false. This implies that $\ell$ is true and $a(\ell), f(\ell) \notin P$ (from Claims~\ref{cl:each clause contains true literal} and~\ref{in positive literal f not in P}, respectively. This contradicts the fact that $P$ reaches $wx$.

%If $\ell'$ is false, If $x$ is True then $a(\ell)b(\ell)$ is a $(-,-)$ edge, and according to Claim~\ref{in positive literal f not in P} $f(\ell) \notin P$, we conclude that $P$ is not malicious.  If $x$ is False, then $\ell'$ is True and from Claim~\ref{cl:each clause contains true literal} $e(\ell')f(\ell')=(-,-)$, we derive a contradiction, since $P'$ is an alternating subpath.

Since $s$ and the $h(\ell)$ with $\ell$ being any negative literal are the only unmatched vertices, we conclude that the condition $(iii)$ is also satisfied, and $M$ is popular.

\section{Proof of Theorem~\ref{thr: 3SAT-popular}}\label{sec:main-results}

\begin{proof}
Suppose first $U^{in}=F^{in}=F^{out} = \emptyset$. Then the problem boils down to deciding if there exists a popular matching that does not contain a given set of nodes. By Lemma~\ref{lem:node-containement}, it is enough to check if this is true for a(ny) stable matchings in $(G,<)$. It is well-known that this can be done in polynomial time, see e.g.~\cite{GI}.

Suppose now $U^{out}=F^{in}=F^{out} = \emptyset$. Then, again by Lemma~\ref{lem:node-containement}, the problem corresponds to deciding if there exists a dominant matching of $(G,<)$ that contains a given set of nodes. It is shown in~\cite{The popular edge problem} how to efficiently construct an instance $(G',<')$ and a surjective map $\sigma$ from stable matchings of $(G',<')$ to dominant matchings of $(G,<)$, with the additional property that a node $v \in V(G)$ is in a dominant matching $\sigma(S)$ of $(G,<)$ if and only if one of certain nodes $v', v''$ are in the stable matching $S$ of $(G',<')$ (see~\cite{The popular edge problem} for details). Since, as mentioned above, whether a node belong to a(ny) stable matching can be checked efficiently, our problem can be solved in polynomial time.

Suppose now $U^{out}=U^{in}=F^{out} = \emptyset$, and $F^{in}=\{e\}$. This is the \emph{popular edge problem}, that is shown in~\cite{The popular edge problem} to be solvable in polynomial time. We call the case when $U^{out}=U^{in}=F^{in} = \emptyset$, and $F^{out}=\{e\}$ the \emph{impopular edge problem}. We show that a solution to the latter follows from the solution to the popular edge problem~\cite{The popular edge problem}. Define $E_s$ and $\bar{E_s}$ as:
	\begin{align*}
	E_s &= \{ e \in E: \exists \text{ stable matching } \M \text{ s.t. } e \in \M\}\\
\bar{E_s} &= \{ e \in E: \exists \text{ stable matching } \M \text{ s.t. } e \notin \M\}
\end{align*}
$E_d$, $\bar{E_d}$ (resp. ${E_p}$, $\bar{E_p}$) are defined similarly, by replacing ``stable'' with ``dominant'' (resp. popular). It is proved in~\cite{The popular edge problem} that $E_p = E_s \cup E_d$. We now argue that $\bar{E_p} = \bar{E_s} \cup \bar{E_d}$.
	
	Consider any $e \in \bar{E_s} \cup \bar{E_d}$. Since there is a stable or dominant matching that does not contain $e$, it follows that  $e \in \bar{E_p}$. We conclude that $\bar{E_p} \supseteq \bar{E_s} \cup \bar{E_d}$. Consider any $e=ij \in \bar{E_p}$. There is a popular matching $P$ s.t. $e \notin P$, and $i$ or $j$ is matched (if $i$ and $j$ are unmatched then $(i,j) = (+,+)$, and $M$ is not popular). Wlog assume $jk \in M$. It follows from $E_p = E_s \cup E_d$ that there exists a stable or dominant matching $M'$ s.t. $jk \in M'$, hence $ij \notin M'$. We conclude that $\bar{E_p} \subseteq \bar{E_s} \cup \bar{E_d}$.
	
	Since $\bar{E_p} = \bar{E_s} \cup \bar{E_d}$, we can solve the forbidden edge problem by checking if $e \in \bar{E_s}$ or $e \in \bar{E_d}$. For stable matchings, this can be done in polynomial time, see e.g.~\cite{GI}. For dominant matchings, this can be done thanks to the mapping $\sigma$ define above, see~\cite{The popular edge problem} for details.

\smallskip

We now move to NP-Complete cases. Testing whether a matching is popular can be performed in polynomial time, see \cite{Biro,Popular matching in the stable marriage problem}. Hence all those problem are in NP. Note that, for a popular matching $M$ of $G(\psi)$, the following are equivalent: $st,wx \notin M$; $tu,xy \in M$;
$s \notin V(M)$, $y \in V(M)$;
$s \notin V(M)$, $yx \in M$;
$y \in V(M)$, $st \notin M$; $y \in V(M)$, $tu \in M$. Hence, Lemma~\ref{M contains (t,u), apex is matched to its first choice} settles all NP-Complete cases with $|U^{in}\cup U^{out} \cup F^{in} \cup F^{out}| = 2$. To show NP-Completeness when sets $U^{out}$ and/or $F^{out}$ have bigger size, we can add nodes $h(\ell)$ to $U^{out}$ and/or evicted edges to $F^{out}$, since we know from Lemma~\ref{M contains (t,u), apex is matched to its first choice} that those nodes and edges are never in a popular matching that contains $tu$. If $U^{in}$ and/or $F^{in}$ have bigger sizes, we can repeatedly add nodes $z,k$, with $z$ adjacent to $k$ and $k$ adjacent to $g(\ell)$ for some literal $\ell$ such that $z,k$ are each other's favorite partner. Then clearly $zk$ belongs to any popular matching, so it can be added to $F^{in}$, and $z$ to $U^{in}$.
\end{proof}

%\begin{remark}
%	Let $H(\psi)$ be the subgraph  on $V(\psi)\setminus \{w,x,y\}$ in the construction in the proof. Notice that a popular matching on $G(\psi)$ is a stable induced matching on $H(\psi)$. Hence, we have also proved that the following problem is also NP-Hard:\\
	
%	Given a preference system $(G,<)$ and nodes $s, t$, find a stable matching such that there is no $s-t$ alternating path in $G_M$. 
%\end{remark}

%Consider first the case $U^{out}=U^{in}=F^{in} = \emptyset$, and $F^{out}=\{e,f\}$. The thesis follows from Lemma~\ref{lem: 3SAT-popular}. If $U^{out}=U^{in}=F^{out} = \emptyset$, and $F^{out}=\{e,f\}$, then note that any popular matching of $G(\psi)$ that does not contain $st$, $wx$ must contain $tu$ and $xy$. The thesis the follows again

%The solution follows from the proof of Lemma~\ref{lem: 3SAT-popular}, since forcing vertex $y$ to be in a matching  and forbidding vertex $s$ from the matching is equivalent to forbidding  $\{st\}$ and $\{wx\}$ edges from the matching. We conclude that this problem is NP-hard. 

\iffalse

\medskip

\noindent \emph{Proof of Theorem~\ref{thr: 3SAT-popular}.} {\color{red}Add it here.}
\hfill $\square$ 

\medskip

\noindent \emph{Proof of Theorem~\ref{thr: 3SAT-popular}.} {\color{red}Add it here.}
\hfill $\square$ 

\medskip

\fi

\section{Consequences of Theorem~\ref{thr: 3SAT-popular}}\label{sec:last}

Lemma~\ref{lem: 3SAT-popular} implies the NP-Hardness of a number of related problems. The first and more important are the hardness of \texttt{mwp} and \texttt{miwp}. We also discuss here how to obtain the $1/2$-approximated algorithm for \texttt{mwp}.

\medskip

\noindent \emph{Proof of Theorem~\ref{thr:mwm}.} Consider the graph $G(\psi)$ from Lemma~\ref{lem: 3SAT-popular}, and give weight $1$ to edges $ut, xy$, and $0$ to all other edges. Suppose we want to solve \texttt{mwp}. Because of Lemma~\ref{lem: 3SAT-popular}, it is NP-Complete to decide whether $G(\psi)$ has a popular matching of weight $2$. In order to prove the positive result, consider the following claim.

\begin{claim}
	Let $M_{st}^*$ and $M_{dom}^*$ be the maximum weight stable and dominant matchings of $(G,\psi)$ respectively. Then $w(M) \leq 2 W$ for any popular matching $M$, where $W = \max(w(M_{st}^*), w(M_{dom}^*))$.
\end{claim}
\begin{proof}
	Assume by contradiction that there exists a popular matching $M$ such that $w(M) > 2W$. It was shown in~\cite{The popular edge problem} that any popular matching $M$ can be partitioned into a stable matching $M_0$ restricted to a subgraph $G'(\psi)$ and a dominant matching $M_1 = M \setminus M_0$ restricted to $G(\psi)\setminus G'(\psi)$, and there always exist matchings $M_0'$ and $M_1'$ such that $M_{st}'=M_0 \cup M_1'$ is stable and $M_{dom}'=M_0' \cup M_1$ is dominant in $G(\psi)$. First assume $w(M_0) \geq w(M_1)$. This implies $w(M_0) = \frac{1}{2} (w(M_0)+w(M_0)) \geq \frac{1}{2} (w(M_0)+w(M_1)) = \frac{1}{2} w(M)$. Consider $M_{st}'=M_0 \cup M_1'$,  $w(M_{st}')=w(M_0) + w(M_1')$. Since weights are non-negative, $w(M_{st}') \geq w(M_0) \geq \frac{1}{2} w(M)$. We conclude that $W \geq w(M_{st}') \geq \frac{1}{2} w(M)$. This contradicts to the assumption that $w(M) > 2W$. If $w(M_1) \geq w(M_0)$, apply the same argument to $M_{dom}'$.
\end{proof} 

The $1/2$-approximation immediately follows from the previous claim, and the fact that a dominant and stable matching of maximum weight can be computed in polynomial time, see again~\cite{The popular edge problem}.

Now change weights so that $w_e=1$ if $e \in \{st,wx\}$, and $0$ otherwise, and suppose we want to solve \texttt{miwp} on the given instance. We deduce from Lemma~\ref{lem: 3SAT-popular} that it is NP-Complete to decide if the optimal solution is $0$. Hence, the problem cannot be approximated to any factor, unless P=NP. 
\hfill $\square$ 

\medskip

Next corollary deals with weighted popular matching problems where weights are given on the \emph{nodes} instead of the edges. 

\begin{corollary} The following problems: \emph{{\bf Given}: a preference system $(G,<)$ with weights $w$ on nodes {\bf Find}: a popular matching $M$ that maximizes (resp. minimizes) $w(M)$} 
	are NP-Hard and not approximable to any factor in polynomial time, unless P=NP. Conversely, if $w \geq 0$, they can be solved in polynomial time. 
\end{corollary}

\begin{proof}
Note that, if we do not make any assumption on $c$, maximizing and minimizing are polynomially equivalent. So we focus on the maximization. Consider the instance $(G,<)$ as defined in Lemma~\ref{lem: 3SAT-popular}, and $w$ being the vector with $-1$ in the component corresponding to $s$, $1$ in the components corresponding to $y$, and $0$ otherwise. Lemma~\ref{lem: 3SAT-popular} implies that it is NP-Complete to decide if the optimum is strictly greater than $0$. On the other hand, if $w\geq 0$, Lemma~\ref{lem:node-containement} implies that the minimum (resp. maximum) of $w(M)$ with $M$ popular is achieved at any stable (resp. dominant) matching, which can be easily computed. \end{proof}

The next problem can be seen as an extension of \texttt{pmffe} when $F^{out}$ is the only non-empty set. 
\begin{corollary}
The following problem: \emph{{\bf Given} a preference system  $(G,<)$ $F\subseteq E(G)$, and $q \in \N$, {\bf Decide:} if there exists a popular matching $M$ of $(G,<)$ such that $|M\cap F|\leq q$} is NP-Complete.
\end{corollary}

\begin{proof}
Clearly the problem is in NP. Consider the preference system $(G,<)$ as defined in Lemma~\ref{lem: 3SAT-popular}. As in the proof of Theorem~\ref{thr: 3SAT-popular}, add $q$ pairs of nodes $k,z$, such that $kz$ is in any popular matching. Let $F$ be the set given by this $q$ pairs, plus $st$, $wx$. Then a popular matching $M$ satisfies $|M\cap F|\leq q$ if and only if $M\cap \{st, wx\}=\emptyset$ if and only if the restricted 3-SAT instance is satisfiable.  \end{proof}

Recall that testing whether a matching is popular can be performed in polynomial time, see e.g. \cite{Biro,Popular matching in the stable marriage problem}. We conclude this section by showing that, on the other hand, deciding whether a subset of nodes can be matched among themselves only as to form a popular matching of the original instance is NP-Complete. We call this the \emph{exclusive popular set problem.} It can also be seen as an extension of \texttt{pmfee} when $U^{in}$ (or $U^{out}$) is the only non-empty set.

\begin{corollary} The following \emph{exclusive popular set problem} (\texttt{eps}): \emph{{\bf Given}: a preference system $(G,<)$ and a set $U\subseteq V(G)$ {\bf Decide}: if there exists a popular matching $M$ of $(G,<)$ with $V(M)=U$} 
	is NP-Complete. \end{corollary}

\begin{proof}
Clearly the problem is in NP. Consider the instance from Lemma \ref{lem: 3SAT-popular}, and set $U=(\{t,u,v,w,x,y\} \cup_\ell V(\ell))\setminus \cup_\ell h(\ell)$, where the union ranges over all literals. Let $M$ be a popular matching of $(G,<)$. We claim that $V(M)=U$ if and only if $st, wx \notin M$. The thesis then follows by Lemma \ref{lem: 3SAT-popular}. Let $V(M)=U$. Then by definition $st \notin M$, and $wx \notin M$, otherwise $y$ is unmatched. Now suppose $st, wx \notin M$. Then $tu \in M$. By Lemma \ref{M contains (t,u), apex is matched to its first choice}, we know that for each literal $\ell$, $g(\ell)f(\ell) \in M$, $h(\ell)$ is $M$-exposed, and $c(\ell)c(\ell')\notin M$ for all $\ell'$. Simple arguments (see e.g. the proof of Lemma \ref{M contains (t,u), apex is matched to its first choice}, part 3), imply that, for a fixed $\ell$, nodes $b(\ell),c(\ell),d(\ell),e(\ell)$ must be all matched, and all among themselves. This implies $V(M)=U$, as required. \end{proof}

\end{document}